\documentclass[10pt]{studiamnew}
\usepackage{graphicx}
\usepackage{amsmath}
\usepackage{amssymb}
\usepackage{url}
\usepackage{tree-dvips}
\usepackage{qtree}
\usepackage{color}

\newtheorem{theorem}{Theorem}[section]
\newtheorem{lemma}[theorem]{Lemma}

\theoremstyle{definition}
\newtheorem{remark}[theorem]{Remark}

\newtheorem{problem}[theorem]{Problem}

\renewcommand{\theequation}{\thesection.\arabic{equation}}
\numberwithin{equation}{section}

\begin{document}
\setcounter{page}{1}
\setcounter{firstpage}{1}
\setcounter{lastpage}{4}
\renewcommand{\currentvolume}{??}
\renewcommand{\currentyear}{??}
\renewcommand{\currentissue}{??}
\title{Fully Subexponential Time Approximation Scheme for Product Partition}
\author{Marius Costandin}
\email{costandinmarius@gmail.com}
%
%
%
\begin{abstract}
In this paper we study the Product Partition Problem (PPP), i.e. we are given a set of $n$ natural numbers represented on $m$ bits each and we are asked if a subset exists such that the product of the numbers in the subset equals the product of the numbers not in the subset.  Our approach is to obtain the integer factorization of each number. This is the subexponential step. We then form a matrix with the exponents of the primes and propose a novel procedure which modifies the given numbers in such a way that their integer factorization contains sufficient primes to facilitate the search for the solution to the partition problem, while maintaining a similar product. We show that the required time and memory to run the proposed algorithm is subexponential. 
\end{abstract}
\maketitle

%
%
%
%
%
%
%
%
%

\section{Introduction}

In this paper the well known Product Partition Problem (PPP) is discussed and some results are obtained for it, namely a subexponential algorithm is proposed. The subexponential step (but not polynomial) is due to factoring of the natural numbers given in the PPP problem. However, the final algorithm calls integer factorization for a subexponential number of integers. Although the integers to be factored are getting exponentially larger with each step we only need an order of $\log(q)$ such steps, where $q$ is the number of distinct primes in the product of the given numbers.  

\section{Main results}

For $n \in \mathbb{N}$, let $S \in \mathbb{N}^{n}$. For $S = \begin{bmatrix} s_1, \hdots, s_n \end{bmatrix}^T$ we assume that $s_i$ is represented on at most $m\in \mathbb{N}$ bits, hence $s_i \leq 2^m$. The Product Partition Problem (PPP) asks:
\begin{problem} \label{P2.1}
 Exists $\mathcal{C} \subseteq \{1, \hdots,n\}$ such that 
\begin{align}\label{E2.1}
\prod_{i\in \mathcal{C}} s_i = \prod_{i \in \{1, \hdots,n\} \setminus \mathcal{C}} s_i \hspace{1cm}?
\end{align} 
\end{problem}

Since interger factoring is known to be subexponential \cite{intF}, the prime factors of the numbers in $S$ can be obtained in subexponential time. Let $\{p_1, \hdots, p_q\}$ be the  set of all the prime numbers involved, in ascending order. As such it is obtained
\begin{align}
s_i = \prod_{k=1}^q p_k^{\alpha_{ik}} \hspace{1cm} \forall i \in \{1, \hdots,n\}
\end{align} where $\alpha_{ik} \in \{0\} \cup \mathbb{N}$.

Construct the Table \ref{Ta1}. 

\begin{table}[h] 
\begin{tabular}{ |p{1cm}|p{1cm}|p{1cm}|p{1cm}|  }
 \hline
 \multicolumn{4}{|c|}{general representation} \\
 \hline
 $ $& $p_1$ & $\hdots$ &$p_q$\\
 \hline
$s_1$&   $\alpha_{11}$ & $\hdots$  &$\alpha_{1q}$\\
 $\vdots$   & $\vdots$    &$\ddots$ &$\vdots$ \\
 $s_n$ &$\alpha_{n1}$ &$\hdots$ &$\alpha_{nq}$\\
 \hline
\end{tabular}
\caption{Matrix representation of the prime powers involved}
\label{Ta1}
\end{table}

From Table \ref{Ta1} define the matrix
\begin{align}\label{E2.3}
S_M = \begin{bmatrix} \alpha_{11} &\hdots &\alpha_{1q}\\ \vdots &\ddots &\vdots\\ \alpha_{n1} &\hdots &\alpha_{nq}\end{bmatrix} \in \mathbb{N}^{n \times q}.
\end{align}

The following result can be stated:
\begin{lemma}\label{L2.1}
The PPP has a solution iff exists $x \in \{0,1\}^n$ such that 
\begin{align}\label{E2.4}
\left(x - \frac{1}{2}\cdot 1_{n\times 1}\right)^T\cdot S_M = 0_{1 \times q}
\end{align}
\end{lemma}
\begin{proof}
In order to have an equal product in (\ref{E2.1}), one needs every prime in each side to have equal power. 

Let $\mathcal{C} \subseteq \{1, \hdots, n\}$ be a solution to PPP, then define $x^T\cdot e_i = 1$ for all $i \in \mathcal{C}$ and zero otherwise. On the other hand, let $x$ be a solution to (\ref{E2.4}), then let $\mathcal{C} = \{ i\in \{1, \hdots, n\} | x^T \cdot e_i = 1\}$. 
\end{proof}

Let $\alpha_{\cdot,k} \in \mathbb{N}^{n}$ denote the column $k$ in $S_M$ for all $k \in \{1, \hdots, q\}$ and  $\alpha_{i,\cdot} \in \mathbb{N}^q$ denote the line $i$ in $S_M$ for all $i \in \{1, \hdots, n\}$. 
\begin{remark}\label{R2.2}
According to Lemma \ref{L2.1} solving the PPP is equivalent with finding $x \in \{0,1\}^n$ with 
\begin{align}
x^T\cdot \alpha_{\cdot,k} = \frac{1}{2}\cdot 1_{n \times 1}^T \cdot \alpha_{\cdot, k} \hspace{0.5cm} \forall k \in \{1, \hdots, q\}
\end{align} i.e. solving simultaneously $q$ Subset Sum Problems where $S_k = \alpha_{\cdot,k}$ and $T_k = \frac{1}{2}\cdot 1_{n \times 1}^T \cdot \alpha_{\cdot,k}$. 
\end{remark}

 However, solving simultaneously $q$ subset sum problems is known to be strongly NP-complete, hence it can not have, in general, a polynomial reduction to the Subset Sum Problem, which is known to be weakly NP-Complete, unless P = NP. 

One can take advantage here on the following observation:
\begin{remark} \label{R2.3}
The elements of $S_M$ are not large! Indeed, since they are actually the exponents of the primes forming each number, if each number is represented on $m$ bits, hence is smaller than $2^m$ and each prime is at least $2$, follows that in particular $\alpha_{ik} \leq m$ for all $i \in \{1, \hdots, n\}$ and $k \in \{1, \hdots, q\}$. 
\end{remark}
\subsection{Primes Pump Algorithm}

In this subsection we assume the PPP has at most one solution, i.e. we consider the Unambiguous Product Partition Problem (UPPP). That is, we search for a solution under the assumption that if it exists, it is unique. This "relaxed" problem it's been shown to be still in NP-Complete. 

However, if the partition problem has a solution, then it has at least two. We can add some polynomial time constraints to the system like $1_{n\times 1}^T \cdot x = N \in \{1, \hdots, n\}$ i.e. the number of ones to be $N$ (which iterates in the set $\{1, \hdots, n\}$). As such, one of the $n$ such problems will have a unique solution, if $n$ is odd. Motivated by this, in the following we search for $x \in \{0,1\}^n$ with $\left( x - \frac{1}{2}\cdot 1_{n\times1} \right)^T \cdot S_M = 0_{q \times 1}$ under the assumption that if it exists it is unique. 

First note that if the number of primes $q = n$, the number of entries in the set $\mathcal{S}$, and it is known that exists a unique $x \in \{0,1\}^n$ such that 
\begin{align}\label{E2.6}
x^T \cdot S_M = \frac{1}{2}\cdot 1_{n\times 1}^T\cdot S_M \iff S_M^T\cdot \left(x - \frac{1}{2}\cdot 1_{n\times 1} \right) = 0_{q \times 1}
\end{align} i.e. $x - \frac{1}{2}\cdot 1_{n\times 1} \in \text{Ker}(S_M^T)$ then, if $\text{dim}\left( \text{Ker}\left(S_M^T\right)\right) = 1$, then it would be easy to find $x$. This would be just a scaled  eigenvector of the $\lambda = 0$ eigenvalue.

However, since for any solution $x$ to the PPP the equation (\ref{E2.6}) is met, it follows that especially if $q < n$, the dimension of the $S_M$ matrix kernel is striclty greater than $1$. Now, given $n \gg q$, we propose an algorithm for modifying the input problem such that in the new problem $n \leq q$ and the new problem has a solution iff the initial problem has one (or at least some approximative solution).

Assume in (\ref{E2.6}) that $S_M^T$ kernel has null dimension, then there is no exact solution to the PPP. But under the assumption that it's eigenvalues are distinct, we can still search for something which is "closest" to a solution as follows: let $\sigma_1 = \lambda_1^2$ be the smallest eigenvalue of $S_M\cdot S_M^T$ and $v_1$ be the corresponding eigenvector with $\|v_1\| = 1$. Take 
\begin{align}
x = \frac{1}{2}\cdot 1_{n\times 1} + \beta \cdot v_1 \hspace{0.5cm} \beta \in \mathbb{R} 
\end{align} and find $\beta$ such that $x$ is the closest to a corner of the unit hypercube. 

This is the reason for which we want more primes in the input.
%

For any $\gamma \in \mathbb{N}$ fixed, it is obvious that $S = \begin{bmatrix} s_1, \hdots, s_n\end{bmatrix}$ PPP has a solution iff $S^{\gamma} = \begin{bmatrix} s_1^{\gamma}, \hdots, s_n^{\gamma}\end{bmatrix}$ PPP has a solution. We define $\hat{s}_i(\gamma)$ as follows. We write $s_i^{\gamma}$ alongside for better comparison.
\begin{align}
s_i^{\gamma} = \prod_{k=1, \alpha_{i,k}\neq 0}^q p_k^{\gamma} \hspace{0.5cm} \hat{s}_i(\gamma) = \prod_{k=1, \alpha_{i,k}\neq 0}^q \left(p_k^{\gamma}-1\right)
\end{align} then $\hat{S}(\gamma) = \begin{bmatrix} \hat{s}_1(\gamma &\hdots &\hat{s}_n(\gamma)\end{bmatrix}$. 

For some $S \in \mathbb{N}^n$ by $PPP(S)$ we refer to the PPP problem associated to $S$. 

\begin{lemma}\label{L2.5}
If $S$ based PPP has a solution $\mathcal{C}$, then for any $\epsilon > 0$ exists $\gamma \in \mathbb{N}$ such that
\begin{align}
\frac{\prod_{i \in \mathcal{C}} \hat{s}_i(\gamma)}{\prod_{i \in \bar{\mathcal{C}}} \hat{s}_i(\gamma)} \leq 1 + \epsilon
\end{align} where $\bar{\mathcal{C}} = \{1, \hdots, n\} \setminus \mathcal{C}$.
\end{lemma}
\begin{proof}
Since it is known that $\mathcal{C}$ is a solution for the PPP(S), i.e. $\frac{\prod_{i \in \mathcal{C}} s_i^{\gamma}}{\prod_{i \in \mathcal{C}} s_i^{\gamma}} = 1$, let us evaluate
\begin{align} 
\frac{\prod_{i \in \mathcal{C}} s_i^{\gamma}}{\prod_{i \in \mathcal{C}} \hat{s}_i(\gamma)} = \prod_{i \in \mathcal{C}} \frac{p_i^{\gamma}}{p_i^{\gamma} - 1} = \prod_{i\in \mathcal{C}} \left( 1 + \frac{1}{p^{\gamma} - 1}\right) = \prod_{i \in \mathcal{C}} \frac{1}{1 - p_i^{-\gamma} } 
\end{align} but

\begin{align}
\frac{1}{\zeta(\gamma)} = \prod_{i=1}^{\infty} \left( 1 - \frac{1}{p_i^{\gamma}} \right) \leq \prod_{i\in \mathcal{C}} \left( 1 - \frac{1}{p_i^{\gamma}} \right) \leq \prod_{i\in \mathcal{C}} \left( 1 + \frac{1}{p_i^{\gamma}-1} \right)
\end{align} and 

\begin{align}
\prod_{i\in \mathcal{C}} \left( 1 + \frac{1}{p_i^{\gamma}-1} \right) = \prod_{i \in \mathcal{C}} \frac{1}{1 - p_i^{-\gamma} } \leq \prod_{i =1}^{\infty} \frac{1}{1 - p_i^{-\gamma} } = \zeta(\gamma)
\end{align} hence

\begin{align}
\frac{1}{\zeta(\gamma)} \leq  \frac{\prod_{i \in \mathcal{C}} s_i^{\gamma}}{\prod_{i \in \mathcal{C}} \hat{s}_i(\gamma)} \leq \zeta(\gamma)
\end{align} A similar result can be obtained for $\bar{\mathcal{C}}$. Then, since $\zeta(\gamma) \to^{\gamma \to \infty} 1$, the conclusion follows.
\end{proof}

We analyze in the following the $PPP(\hat{S}(\gamma))$ problem. It is natural to ask: will the number of primes in the factorization of $\hat{s}_i(\gamma)$ for all $i\in \{1, \hdots, n\}$ be greater or smaller than the number of primes in the factorization of $s_i$ for $i \in \{1, \hdots, n\}$. We want therefore to analyze
\begin{align}
\omega\left( \prod_{i =1}^n s_i^{\gamma} \right) = q \hspace{0.5cm}\omega\left(\prod_{i =1}^n \hat{s}_i(\gamma)\right) \geq^{?} q
\end{align} where $\omega(\cdot)$ is the little prime omega function. 

For this analysis, we restrict ourselves again and assume $p_1, \hdots, p_q$ in the $PPP(S)$ problem are the first $q$ consecutive primes. Under this assumption we give the following result:
\begin{lemma} \label{L2.6}
For $q \in \mathbb{N}$, let us denote 
\begin{align}
T_q = \prod_{k=1}^q \left(p_k^\gamma - 1\right)
\end{align} then
\begin{align}
\Omega\left( T(q)\right) \sim \sigma_0(\gamma) \cdot q \cdot \log\left( \log \left( q\right)\right)
\end{align} where $\Omega(\cdot)$ is the big prime omega function and $\sigma_0(\gamma) = \sum_{d|\gamma} 1 $ is the divisor counting function.
\end{lemma} We thank the MathOverflow user Ofir Gorodetsky for sketching the proof. 
\begin{proof}
We first start with a result due to H. Halberstam in \cite{Halbe}. For an ireducible polynomial $f$, one has
\begin{align}
\sum_{p \leq q} \omega(f(p)) \sim \frac{q}{\log(q)}\cdot \log(\log(q))
\end{align} as $q \to \infty$.  We have

\begin{align}\label{E2.18}
\Omega(T_q) = \Omega\left( \prod_{k=1}^q \left(p_k^\gamma - 1\right)\right) = \sum_{k =1}^q \Omega\left( p_k^\gamma - 1 \right)
\end{align} Although $p_k^{\gamma} - 1$ is not irreducible, it factors into $\phi(\gamma)$ cyclotomic ireducible polynomials
\begin{align}
x^{\gamma} - 1 = \prod_{d | \gamma} \Phi_d(x)
\end{align} As such

\begin{align}
\Omega(T_q) &= \sum_{k=1}^q  \sum_{d | \gamma} \Omega\left( \Phi_d(p_k) \right)  = \sum_{d | \gamma} \left(  \sum_{k=1}^q \Omega\left( \Phi_d(p_k) \right) \right)
\end{align} 

Next, if $\sum_{k = 1}^q \Omega(\Phi_d(p_k)) \sim  \sum_{k=1}^q \omega(\Phi_d(p_k))$ \color{red} \textbf{this needs attention} \color{black} we obtain
\begin{align}\label{E2.21}
\sum_{p \leq p_q} \Omega(\Phi_d(p)) \sim \frac{p_q}{\log(p_q)}\cdot \log(\log(p_q))
\end{align}
From (\ref{E2.18}) and (\ref{E2.21}) one gets
\begin{align}
\Omega(T_q) \sim \left(\sum_{d|\gamma} 1 \right) \cdot \frac{p_q}{\log(p_q)}\cdot \log(\log(p_q))
\end{align} Finally since $p_q \sim q \cdot \log(q)$ we get
\begin{align}
\Omega(T_q) \sim \sigma_0\left(\gamma \right)  \cdot q \cdot \log \log q
\end{align}

\end{proof}

Taking $\gamma$ to be some prime number, we get

\begin{align}
\Omega\left( \prod_{k=1}^q \left(p_k^{\gamma} - 1\right)\right) \sim 2 \cdot q \cdot \log\log q
\end{align} whereas 
\begin{align}
\Omega\left( \prod_{k=1}^q p_k^{\gamma}\right) = \gamma \cdot q 
\end{align} For large values of $q$ one has $\gamma \ll 2 \cdot \log\log(q)$ hence this step produces more prime numbers, hence more columns in the matrix $S_M$ in (\ref{E2.3}). We call this procedure "pumping primes procedure". It is summarized below:
\begin{enumerate}
\item Choose $\gamma > 1$ prime such that the product $\prod p^\gamma$ is approximated by the product $\prod \left(p^{\gamma} - 1\right)$ as shown in Lemma \ref{L2.5}. 
\item For each prime we obtain the integer factorization of $p^{\gamma}-1$ and we form a new set of entries as follows. If for instance $s_1 = p_1 \cdot p_2$ then $s_1^{\gamma} = p_1^{\gamma} \cdot p_2^{\gamma}$. Assuming $p_1^{\gamma} - 1 = p_{1,1}\cdot p_{1,2}$ and $p_2^{\gamma} - 1 = p_{2,1}\cdot p_{2,2} \cdot p_{2,3}$, we have the entry in the new list $s_{1,1}(\gamma) = p_{1,1}\cdot p_{1,2} \cdot p_{2,1}\cdot p_{2,2}\cdot p_{2,3}$. As such we obtain a new matrix $S_{M,1}(\gamma)$ with the same number of lines, but more columns (under the hypethesis that $\omega(T_q)$ increases with the same rate as $\Omega(T_q)$)  
\item Let $a \in \mathbb{N}$ be a natural number and assume that $n \sim q^a$. We want to apply again the previous step to obtain more prime numbers, hence a new matrix $S_{M,2}(\gamma)$ with more columns and the same number of lines. 
\end{enumerate}

Apply again the above explained step is motivated as follows: althought the obtained primes may not be consecutive, the result from Lemma \ref{L2.6} shows something on the lines of  
\begin{align}
\Omega\left( \prod_{p \in \mathcal{P}} \left( p^{\gamma} - 1 \right) \right) \sim \sigma_0(\gamma) \cdot |\mathcal{P}| \cdot \log\left( \log\left( |\mathcal{P}|\right)\right)
\end{align} where $\mathcal{P}$ is a set containing primes $|\mathcal{P}| $ is the number of elements in $\mathcal{P}$ and $\sigma_0(\gamma)$ counts the divisors of $\gamma$.

Applying the above procedure for $K$ times one gets the number of primes for $\gamma$ a prime as well (hence $\sigma_0(\gamma) = \sum_{d|\gamma}1 = 2$) in the order
\begin{align}
2^K \cdot q\cdot \log\left( \log \left( q \right)\right)^K \sim q^a
\end{align} hence the number of application is given by:
\begin{align}
K\cdot \log \left( 2 \cdot \log(\log(q))\right) + \log(q) \sim a \cdot \log(q)
\end{align} which is
\begin{align}
K \sim \frac{(a-1) \cdot \log(q)}{\log(2) + \log(\log(\log(q)))}
\end{align}

Let $\mathcal{P}_K$ denote the set of distinct primes obtained after the application of the above presented algorithm. We have
\begin{align}
\frac{1}{\zeta(\gamma)^K} \leq \frac{\prod_{p\in \mathcal{P}_0} p^{\gamma}}{\prod_{p \in \mathcal{P}_K} \left(p^{\gamma}-1\right)} \leq \zeta(\gamma)^K
\end{align}

At each step, we should expect the new primes $p'$ to be in the range $p^{\gamma} - 1$ where $p$ is an old prime. As such, after $K$ steps one gets the size of the last primes in the range $p^{\gamma^K}$ which require the amount of memory
\begin{align}
 \log\left(p^{\gamma^K}\right) \sim \gamma^{K}\cdot \log(p).
\end{align} This shows a subexponential increase in the needed memory for the numbers involved since $\gamma$ is considered constant.

%
%

\section{Conclusion and future work}
 
A sketch for a Fully Subexponential Time Approximation Scheme (FSTAS) was presented for the Product Partition problem. The approach was to obtain the integer factorization of the numbers in the given set, then through a novel procedure to obtain a new  set of numbers which have a "richer" factorization, but retain a close value to the initial numbers. There are two directions which need to be studied further:
\begin{enumerate}
\item the number of primes resulting after the "prime pump procedure" is evaluated using the big omega function $\Omega(\cdot)$, while is the little omega function $\omega(\cdot)$ the one which adds columns in $S_M$ and is thus of interest to us. However, we were not able to properly estimate the values of $\omega(\cdot)$
\item After each "prime pump procedure" step, new primes are obtained which form the columns in the matrix $S_{M,K}(\gamma)$ (see \ref{E2.3} for the definition of the matrix). We were not able to show that the matrix rank indeed increases (even if we assume that $\omega(\cdot) \sim \Omega(\cdot)$). A proper analysis of the rank of this matrix is needed. 
\end{enumerate}

\end{document}